\documentclass[a4paper]{article}
\usepackage{graphicx}
\usepackage{url} 
\usepackage{hyperref}
\usepackage{a4wide}
\usepackage{amsmath,amsfonts,amssymb,amsthm}
\newcommand{\R}{{\mathbb R}}

\newtheorem{theorem}{Theorem}[section]

\newtheorem{proposition}[theorem]{Proposition}

\newtheorem{lemma}[theorem]{Lemma}

\theoremstyle{definition}

\numberwithin{equation}{section}

\begin{document}

\noindent 
\begin{center}
\textbf{\large MIMO capacity for deterministic channel models: sublinear growth}
\end{center}

\begin{center}
November 6, 2011
\end{center}

\vspace{0.5cm}

\noindent 

\begin{center}
\textbf{ 
Fran{\c c}ois Bentosela\footnote{Centre de Physique Th{\' e}orique, Aix-Marseille Univ, CNRS UMR 6207, 
13288 Marseille Cedex 9, France}, 
Horia D. Cornean\footnote{Department of Mathematical Sciences, 
    Aalborg
    University, Fredrik Bajers Vej 7G, 9220 Aalborg, Denmark},
Nicola Marchetti\footnote{CTVR/The Telecommunications Research Centre, Trinity College, Dublin 2, Ireland}
   }

\end{center}

\noindent

\begin{abstract}
This is the second paper of the authors 
in a series concerned with the development of a deterministic model for the transfer matrix of a MIMO system. 
Starting from the Maxwell equations, we have described in \cite{BCFM}
the generic structure of such a deterministic transfer matrix. In the current paper 
we apply the results of \cite{BCFM} in order to study the
(Shannon-Foschini) capacity behavior of a MIMO system as a function of the
deterministic spread function of the environment, and the number of transmitting  
and receiving antennas. The antennas are assumed to fill in a given, fixed volume. 
Under some generic assumptions, we prove that 
the capacity grows much more slowly than linearly with the number of
antennas. These results reinforce previous heuristic results obtained
from statistical models of the transfer matrix, which also predict a sublinear behavior.
\end{abstract}

\vspace{0.5cm}


\vspace{0.5cm}


\vspace{0.5cm}

\section{Introduction and the main result}
Multiple Input Multiple Output (MIMO) is a powerful technology for increasing data rates in wireless telecommunication. 
Experimental and theoretical studies show the increase of the capacity 
(number of bits transmitted per second) when the number of the transmitting ($TX$) and receiving ($RX$) antennas also increases. 
Recall (see \cite{Fo, Te}) that when the channel is unknown to the transmitter, the Shannon-Foschini capacity is given by  
\begin{equation}\label{capa}
C(M_T,M_R):=\log_2 \;{\rm det} \Big(\mathbb{I}_{M_R}+\frac{E_T}{M_T N_0}\mathbb{H}\mathbb{H}^*\Big),
\end{equation} 
where $M_T$ is the number of $TX$ antennas, $M_R$ is the number of $RX$ antennas, $\mathbb{I}_{M_R}$ is the $M_R\times M_R$ 
identity matrix, $E_T$ is the average total energy transmitted by the  $TX$ antennas, $N_0$ is the variance of the noise, 
$\mathbb{H}$ is the $M_R\times M_T$ channel  transfer matrix which establishes the linear relationship between the signals at 
$RX$ antennas and the signals at the $TX$ antennas.

Starting from the Maxwell equations, we have shown in \cite{BCFM} what is the generic structure of such a transfer matrix (see below \eqref{iulie11-1}). 
In the present 
paper we apply the formula obtained in \cite{BCFM} and study the behavior of the MIMO capacity as a function of $M_T$, $M_R$ and of the
deterministic spread function of the environment. The antennas are assumed to fill in a given, fixed volume. 
According to \cite{BCFM}, under certain conditions the transfer matrix can be well approximated by: 
\begin{equation}\label{iulie11-1}
\mathbb{H}=\int_{\mathbb{S}^2\times \mathbb{S}^2}d\Omega_R d\Omega_T \;a_R(\Omega_R) s(\Omega_R,\Omega_T) a_T(\Omega_T),
\end{equation}
where $\mathbb{S}^2$ is the two dimensional sphere, $s(\Omega_R,\Omega_T)$ is a $6\times 6$ matrix called the spread 
function which contains the scattering information of the environment, 
$$a_T(\Omega_T):=\{a_T^{im}(\Omega_T)\in \mathbb{C}: i\in\{1,2,\dots, 6\}, m\in\{1,2,\dots, M_T\}\}$$
is a $6\times M_T$ matrix valued map which describes the radiation pattern of the transmitting system, while 
$$a_R(\Omega_R):=\{a_R^{mj}(\Omega_R)\in \mathbb{C}:  m\in\{1,2,\dots, M_R\},j\in\{1,2,\dots, 6\}\}$$
is a $M_R\times 6$ matrix valued map which describes the receiving system. We assume that all 
$a_T^{im}$ and $a_R^{mj}$ are continuous functions of the angles.

The index $m$ accounts for the placement of the $m$'th transmitting/receiving antenna. 
Assuming that the transmitting/receiving antennas are placed in a finite volume $V_{T/R}$, 
the distance between them becomes smaller. 
Moreover, reasoning in terms of Riemann sums, we will assume that there exist two 
smooth kernels $\mathcal{A}_{T/R}^{ij}(\Omega,\Omega')$ such that:
\begin{align}\label{iulie11-12}
&\sup_{\Omega,\Omega'\in \mathbb{S}^2}\left |\mathcal{A}_{T}^{ij}(\Omega,\Omega')-
\frac{1}{M_{T}}\sum_{m=1}^{M_{T}}a_{T}^{im}(\Omega)\overline{a_{T}^{jm}(\Omega')}\right 
|=\mathcal{O}(M_{T}^{-1}),\nonumber \\
&\sup_{\Omega,\Omega'\in \mathbb{S}^2}\left |\mathcal{A}_{R}^{ij}(\Omega,\Omega')-
\frac{1}{M_{R}}\sum_{m=1}^{M_{R}}\overline{a_{R}^{mi}}(\Omega)a_{R}^{mj}(\Omega')\right 
|=\mathcal{O}(M_{R}^{-1}),\quad i,j\in\{1,2,\dots,6\}.
\end{align}

And here is our main result:

\begin{theorem}\label{theo} Assume that the spread function $s(\Omega_R,\Omega_T)$ is the kernel of a Hilbert-Schmidt operator. 
We have the following situations: 

\noindent {\rm (i)}. If $M_R$ is fixed, then $$\lim_{M_T\to\infty} C(M_T,M_R) <\infty;$$ 

\noindent {\rm (ii)}. If $M_T$ is fixed, then  
$$ \lim_{M_R\to\infty}\frac{C(M_T,M_R)}{\ln(M_R)}\leq \frac{M_T}{\ln(2)} ;$$

\noindent {\rm (iii)}. Let $M=M_T$ and assume that $M_R=aM$ for some constant $a>0$. 
If the operator generated by the spread function of the environment 
has finite rank $N<\infty$, then:
\begin{equation}\label{iulie11-13}
\limsup_{M\to\infty}\frac{C(M,aM)}{\ln(M)}\leq \frac{N}{\ln(2)} .
\end{equation}

\noindent {\rm (iv)}. Let $M=M_T$ and assume that $M_R=aM$ for some constant $a>0$. If the 
spread function is $C^\infty$ in both angular 
variables, then for every $\epsilon>0$ we have:
\begin{equation}
\lim_{M\to\infty}\frac{C(M,aM)}{M^\epsilon}=0.
\end{equation}

\end{theorem}

\vspace{0.5cm}

\noindent{\bf Remark 1}. Let us go through some of the previous results obtained with probabilistic models for the channel transfer matrix. 
In the case when the distance between antennas is kept constant, some theoretical studies \cite{P, M, Al} conclude that the 
capacity grows {\it linearly} with the number of antennas. Still for probabilistic models, if the antennas are forced to occupy a fixed volume 
\cite{G1, W} then one has to consider correlations between them. This is done by introducing some ad-hoc 
correlation matrices depending on the interelement distances. In this case, they observe that the capacity either grows at most like a 
{\it logarithm} 
\cite{W}, or even converges to some {\it finite value} \cite{G1}. In \cite{Po} it is also suggested that the capacity should tend 
to a limit as the number of antennas increases in a fixed volume. 

\vspace{0.5cm}

\noindent{\bf Remark 2}. In contrast with the probabilistic models where the correlations between antennas are introduced rather arbitrarily, 
in \cite{BCFM} we developped a deterministic ab-initio model for the channel transfer matrix (see \eqref{iulie11-1}) 
which implicitely takes into account these correlations, through the matrices $a_T$ and $a_R$ which completely describe the radiation patterns of the 
transmitting and receiving arrays, while the spread function $s(\Omega,\Omega')$ describes the scattering environment.  

The mathematical technical assumptions in Theorem \ref{theo} are physically natural, thus our results confirm that the capacity of a system in a 
realistic environment grows more slowly than linearly. For example, in {\it (i)}  we obtain 
that the capacity saturates when $M_T$ grows and $M_R$ is kept fixed; 
the physical explanation is that the spread function cannot convey enough transmit spatial diversity to the receiving side. Similarly, 
if $M_T$ is kept fixed as in {\it (ii)}, there is not enough transmit spatial diversity to start with  
and the capacity only increases as $\ln(M_R)$. When 
both $M_T$ and $M_R$ grow proportionally at the same time, then if there is not enough spatial diversity in the scattering environment as it happens in 
{\it (iii)}, we again only get a logarithmic growth. Finally, when both $M_T$ and $M_R$ grow and the spread function is varying 
smoothly, the growth is slower than any positive power of $M$. 

\vspace{0.5cm}

\noindent{\bf Remark 3}. Two important parameters which implicitely appear in {\it (iii)} and {\it (iv)} are on one hand 
the value of the rank, and on the other hand the speed of oscillations of the spread function. Our proofs implicitely 
show that these factors are maybe more important in the capacity growth than the number of antennas. For the same distribution 
of antennas, the capacity should be larger if the environment contains a lot of scatterers and the spread 
function is very irregular.  

\vspace{0.5cm}

The structure of our paper is as follows: in Section \ref{section1} we express the capacity as a 
Fredholm determinant of an integral operator whose integral kernel depends on $ M_T$ and $M_R$ in a way which 
is easier to deal with when these numbers grow. 
In  Section \ref{section2} we prove that the capacity saturates as a function of $M_T$,  
while it can grow with $M_R$ either logarithmically or power-like, but with arbitrarily small exponents.

\section{Shannon-Foschini capacity as a Fredholm determinant}\label{section1}
 
The main result of this section is contained in Proposition \ref{cor1}, but we need to start with a few technical lemmas.
\begin{lemma}\label{lemma1}
 Introduce the bounded linear maps $A_{T/R}\in B([L^2(\mathbb{S}^2)]^6)$ defined by ($\Psi\in [L^2(\mathbb{S}^2)]^6$):
\begin{equation}\label{iulie11-4}
(A_{T}\Psi)(\Omega):=a_{T}(\Omega)
\int_{\mathbb{S}^2} d\Omega' \;a^*_{T}(\Omega')\Psi(\Omega'),\quad (A_{R}\Psi)(\Omega):=a_{R}^*(\Omega)
\int_{\mathbb{S}^2} d\Omega' \;a_{R}(\Omega')\Psi(\Omega').
\end{equation}
Then $A_{T/R}$ are self-adjoint and non-negative operators.
\end{lemma}
\begin{proof}
Choose $\Psi=\{\psi_i\}_{i=1}^6\in [L^2(\mathbb{S}^2)]^6$ and $\Phi=\{\phi_i\}_{i=1}^6\in [L^2(\mathbb{S}^2)]^6$ and compute:
\begin{align}\label{iulie11-55}
\langle\Psi,A_T\Phi\rangle&
=\int_{\mathbb{S}^2} d\Omega \sum_{i=1}^6\overline{\psi_i(\Omega)} \sum_{m=1}^{M_T}a_T^{im}(\Omega)
\sum_{j=1}^6\int_{\mathbb{S}^2} d\Omega'\; \overline{a_T^{jm}(\Omega')}\phi_{j}(\Omega')\nonumber \\
&=\sum_{m=1}^{M_T} \overline{\left\{\sum_{j=1}^6\int_{\mathbb{S}^2}
 d\Omega\; \overline{a_T^{jm}(\Omega)}\psi_j(\Omega)\right\} }\left\{\sum_{j=1}^6\int_{\mathbb{S}^2}
 d\Omega\; \overline{a_T^{jm}(\Omega)}\phi_j(\Omega)\right\}
\end{align}
and observe that $\langle\Psi,A_T\Phi\rangle= \overline{\langle\Phi,A_T\Psi\rangle}=\langle A_T\Psi,\Phi\rangle$. Moreover, 
\begin{align}\label{iulie11-5}
\langle\Psi,A_T\Psi\rangle=\sum_{m=1}^{M_T}\left |\sum_{j=1}^6\int_{\mathbb{S}^2}
 d\Omega\; \overline{a_T^{jm}(\Omega)}\psi_j(\Omega)\right|^2\geq 0.
\end{align}
 The proof for $A_R$ is similar. 
\end{proof}

\vspace{0.5cm}

Now let us consider $B=\mathbb{H}\mathbb{H}^*$ (the $M_R\times M_R$ 
matrix appearing in the capacity formula \eqref{capa}) and compute using \eqref{iulie11-1}:
\begin{align}\label{iulie11-3}
 &B=\mathbb{H}\mathbb{H}^*=\int_{\mathbb{S}^2\times \mathbb{S}^2} 
d\Omega_R d\Omega'_R\; a_R(\Omega_R)b(\Omega_R,\Omega'_R)a^*_R(\Omega'_R), \quad {\rm where}\nonumber \\
&b(\Omega_R,\Omega_R'):=\int_{\mathbb{S}^2\times \mathbb{S}^2} d\Omega_T d\Omega'_{T} s(\Omega_R,\Omega_T)\;a_T(\Omega_T)a_T^*(\Omega_T')
s^*(\Omega'_R,\Omega'_T).
\end{align}

We denote by $\mathcal{B}$ and $\mathcal{S}$ the integral operators in $B([L^2(\mathbb{S}^2)]^6)$ given by the 
matrix valued kernels $b(\Omega,\Omega')$ and $s(\Omega,\Omega')$ respectively. 
Note that the adjoint of $\mathcal{S}$ in  $B([L^2(\mathbb{S}^2)]^6)$, denoted by $\mathcal{S}^*$, 
will have an integral kernel $\mathcal{S}^*(\Omega,\Omega')=s^*(\Omega',\Omega)$. Thus from \eqref{iulie11-3} we have 
\begin{align}\label{iulie11-6}
\mathcal{B}=\mathcal{S}A_T \mathcal{S}^*.
\end{align}

\begin{lemma}\label{lemma2}
Let $\mathcal{K}:=\sqrt{A_T}\mathcal{S}^*A_R\mathcal{S}\sqrt{A_T}=
\left (\sqrt{A_R}\mathcal{S}\sqrt{A_T}\right)^*\left (\sqrt{A_R}\mathcal{S}\sqrt{A_T}\right)\geq 0$ be a 
bounded positive operator in $B([L^2(\mathbb{S}^2)]^6)$. Then for any integer $k\geq 2$ we have 
$$B^k=\int_{\mathbb{S}^2\times \mathbb{S}^2} d\Omega d\Omega'\; a_R(\Omega)b_k(\Omega,\Omega')a_R^*(\Omega'),$$
where $b_k(\Omega,\Omega')$ is the integral kernel of the operator 
$\mathcal{B}_k=\mathcal{S}\sqrt{A_T}\mathcal{K}^{k-1}\sqrt{A_T}\mathcal{S}^*$ in $B([L^2(\mathbb{S}^2)]^6)$.
\end{lemma}
\begin{proof}
Let us first show that the identity holds for $k=2$. We have:
\begin{align}\label{iulie11-7}
 B^2
&=\int_{\mathbb{S}^2\times \mathbb{S}^2} d\Omega_1 d\Omega_4\;a_R(\Omega_1)
\left (\int_{\mathbb{S}^2\times \mathbb{S}^2} d\Omega_2 d\Omega_3\; b(\Omega_1,\Omega_2)a_R^*(\Omega_2)a_R(\Omega_3)b(\Omega_3,\Omega_4)
\right ) a_R^*(\Omega_4)\nonumber \\
 &=\int_{\mathbb{S}^2\times\mathbb{S}^2} d\Omega d\Omega'\;a_R(\Omega)b_2(\Omega,\Omega')a_R^*(\Omega')
\end{align}
where $b_2(\Omega_R,\Omega'_R)$ is the integral kernel of the operator 
$$\mathcal{B}_2= \mathcal{B}A_R \mathcal{B}=\mathcal{S}A_T \mathcal{S}^*A_R\mathcal{S}A_T \mathcal{S}^*=
\mathcal{S}\sqrt{A_T}\mathcal{K}\sqrt{A_T}\mathcal{S}^*
,$$ which proves the case $k=2$. For $k\geq 2$ we have 
\begin{align}\label{iulie11-8}
 B^{k+1}&=B^kB \nonumber \\
&=\int_{\mathbb{S}^2\times \mathbb{S}^2} d\Omega_1 d\Omega_4\;a_R(\Omega_1)
\left (\int_{\mathbb{S}^2\times \mathbb{S}^2} d\Omega_2 d\Omega_3\; b_k(\Omega_1,\Omega_2)a^*_R(\Omega_2)a_R(\Omega_3)b(\Omega_3,\Omega_4)
\right ) a_R^*(\Omega_4)\nonumber \\
 &=\int_{\mathbb{S}^2\times \mathbb{S}^2} d\Omega d\Omega'\;a_R(\Omega)b_{k+1}(\Omega,\Omega')a_R^*(\Omega')
\end{align}
where $b_k(\Omega,\Omega')$ is the integral kernel of the operator 
$\mathcal{B}_k$, and $b_{k+1}(\Omega,\Omega')$ is the integral kernel of

$$\mathcal{B}_{k+1}=\mathcal{B}_kA_R\mathcal{B}=
(\mathcal{S}\sqrt{A_T}\mathcal{K}^{k-1}\sqrt{A_T}\mathcal{S}^*)
A_R(\mathcal{S}A_T \mathcal{S}^*)=\mathcal{S}\sqrt{A_T}\mathcal{K}^k\sqrt{A_T}\mathcal{S}^*.$$ 
The proof is over.
\end{proof}

\vspace{0.5cm}

\begin{lemma}\label{lemma3}
 Let $z$ belong to the intersection of the resolvent sets of $B$ and $K$, i.e. $z\in \rho(B)\cap \rho(\mathcal{K})$. Then $\mathcal{K}$ is trace class and we have the identity:
\begin{align}\label{iulie11-9}
{\rm Tr}_{\mathbb{C}^{M_R}}\left \{(z-B)^{-1}B\right \}={\rm Tr}_{[L^2(\mathbb{S}^2)]^6}\left \{(z-\mathcal{K})^{-1}\mathcal{K}\right \}.
\end{align}
\end{lemma}
\begin{proof}
 The operator $\mathcal{K}$ is a product of two Hilbert-Schmidt operators, thus it is trace class. Since both sides of \eqref{iulie11-9} 
are analytic on $\rho(B)\cap \rho(\mathcal{K})$, it is enough to prove the equality for $|z|>\max\{||B||, ||\mathcal{K}||\}$. Using the 
power series expansion we can write:
$${\rm Tr}_{\mathbb{C}^{M_R}}\left ((z-B)^{-1}B\right )=
\sum_{n\geq 1}\frac{1}{z^{n}}{\rm Tr}_{\mathbb{C}^{M_R}}B^{n}.$$
Then, using Lemma \ref{lemma2} and the trace cyclicity we get: 
 \begin{align}\label{iulie11-10}
\sum_{n\geq 1}\frac{1}{z^{n}}{\rm Tr}_{\mathbb{C}^{M_R}}B^{n}&=\sum_{n\geq 1}\frac{1}{z^{n}}{\rm Tr}_{[L^2(\mathbb{S}^2)]^6}
\{\mathcal{B}_{n}A_R\}\nonumber \\
&=\sum_{n\geq 1}\frac{1}{z^{n}}{\rm Tr}_{[L^2(\mathbb{S}^2)]^6}
\{\mathcal{S}\sqrt{A_T}\mathcal{K}^{n-1}\sqrt{A_T}\mathcal{S}^*A_R\}=
\sum_{n\geq 1}\frac{1}{z^{n}}{\rm Tr}_{[L^2(\mathbb{S}^2)]^6}\left \{\mathcal{K}^{n}\right\}\nonumber \\
&={\rm Tr}_{[L^2(\mathbb{S}^2)]^6}\left \{(z-\mathcal{K})^{-1}\mathcal{K}\right \},
\end{align} 
\end{proof}

\vspace{0.5cm}

\begin{proposition}\label{cor1}
 The capacity of our system can be written as:
$$C(M_T,M_R)=\frac{1}{\rm ln(2)} {\rm Tr}_{[L^2(\mathbb{S}^2)]^6}{\rm Ln} \left (1+\frac{E_T}{ M_T N_0}\mathcal{K}\right ),$$
where $\mathcal{K}=\sqrt{A_T}\mathcal{S}^*A_R\mathcal{S}\sqrt{A_T}$, while 
${\rm Ln}$ denotes the principal branch of the natural logarithm.
\end{proposition}
\begin{proof}
The Shannon-Foschini capacity is equal to  
$$C(M_T,M_R)=\frac{1}{\rm ln(2)} {\rm Tr}_{\mathbb{C}^{M_R}}
{\rm Ln}\left (1+\frac{E_T}{N_0 M_T}B\right ).$$
Since the spectra of  $\frac{E_T}{N_0 M_T}B$ and $\frac{E_T}{N_0 M_T}\mathcal{K}$ 
are both included in some large enough closed interval of the type $[0,L]$, we can find a positively 
oriented, simple and closed contour $\Gamma$ which contains $[0,L]$, while 
$\Gamma\subset \mathbb{C}\setminus (-\infty,-N_0 M_T/E_T]$. By the Dunford-Riesz functional calculus we may write:
$$
{\rm Ln}\left (1+\frac{E_T}{N_0 M_T}B\right )=\frac{1}{2\pi i} \int_\Gamma dz \;
{\rm Ln}\left (1+\frac{E_T}{N_0 M_T}z\right )(z-B)^{-1}$$
Since the function $g(z):=\frac{1}{z}{\rm Ln}\left (1+\frac{E_T}{N_0 M_T}z\right )$ is also holomorphic on 
$\mathbb{C}\setminus (-\infty,-N_0 M_T/E_T]$, we may write:
$$
{\rm Ln}\left (1+\frac{E_T}{N_0 M_T}B\right )=\frac{1}{2\pi i} \int_\Gamma dz \; 
g(z)(z-B)^{-1}z=\frac{1}{2\pi i} \int_\Gamma dz \; g(z)(z-B)^{-1}B.$$
Thus (using \eqref{iulie11-9}):

$$C=\frac{1}{\rm ln(2)} \frac{1}{2\pi i} \int_\Gamma dz \; g(z)\;{\rm Tr}_{\mathbb{C}^{M_R}}\left\{(z-B)^{-1}B\right\}=
\frac{1}{\rm ln(2)} \frac{1}{2\pi i} \int_\Gamma dz \; g(z)\;{\rm Tr}_{[L^2(\mathbb{S}^2)]^6}
\left \{(z-\mathcal{K})^{-1}\mathcal{K}\right \}$$
which can be integrated back and we obtain the result.
\end{proof}

\vspace{0.5cm}

\section{Proof of Theorem \ref{theo}}\label{section2}
We start with an abstract technical lemma which will be used extensively during the proof. 
\begin{lemma}\label{lemmasept}
Let $F(x)={\rm Ln}(1+x)$ be defined on $\mathbb{C}\setminus (-\infty,-1]$. Let $T_1$ and $T_2$ be two self-adjoint non-negative trace class 
operators defined on a separable Hilbert space $\mathcal{H}$. Then 
$F(T_{1,2})$ are trace class and moreover, if $\Delta T:=T_2-T_1$ we have:
\begin{align}\label{septemb1}
&|{\rm Tr} \{F(T_2)\}-{\rm Tr} \{F(T_1)\}|\leq ||\Delta T||_1,\quad {\rm and}\\
\label{septemb2}
&{\rm Tr} \{F(T_2)\}\leq {\rm Tr} \{F(T_1)\}+{\rm Tr} \{F'(T_1)\Delta T\}.
\end{align}
\end{lemma}
\begin{proof}
 If $0\leq t\leq 1$ we define $T(t):=(1-t)T_1+tT_2$. Clearly, $T(t)\geq 0$ and
$T'(t)=\Delta T\in B_1(\mathcal{H})$ (the space of trace class operators). 
Let $\Gamma$ a positively oriented
simple contour contained in the analyticity domain of $F$, and
surrounding the interval
$[0,\max\{\sigma(T_1),\sigma(T_2)\}]$. Then $\Gamma$
completely contains the spectrum of $T(t)$ for all $t$ and we can
write (in the sense of bounded operators): 
\begin{align}\label{septemb3}
F(T(t))=\frac{1}{2\pi i}\int_\Gamma dz \; F(z)
(z-T(t))^{-1}=\frac{1}{2\pi i}\int_\Gamma dz \; F(z)
\{(z-T(t))^{-1}-1/z\}.
\end{align}
The second formula holds true because $F(z)/z$ is still analytic inside the domain of integration (see also the 
argument used in Proposition \ref{cor1}). But $(z-T(t))^{-1}-1/z$ is now a trace class operator and it follows that 
$F(T(t))\in B_1(\mathcal{H})$.  

Define the function 
$\phi:[0,1]\to\R$ given by: 
\begin{equation}\label{iulie2}
\phi(t):={\rm Tr}
\{F(T(t))\}=\frac{1}{2\pi i}\int_\Gamma F(z){\rm Tr}
\{(z-T(t))^{-1}-1/z\}dz.
\end{equation}
 Let us first compute $\phi'(t)$. Using the cyclicity of the trace we have:
\begin{align}\label{iulie3}
\phi'(t)&=\frac{1}{2\pi i}\int_\Gamma F(z){\rm Tr}
\{(z-T(t))^{-1}T'(t)(z-T(t))^{-1}\}dz\nonumber \\
&=\frac{1}{2\pi i}\int_\Gamma F(z){\rm Tr}
\{(z-T(t))^{-2}T'(t)\}dz.
\end{align}
Denote by $\{|f_j(t)\rangle\}_{j\geq 1}$ the orthonormal basis consisting of eigenvectors of $T(t)$, 
corresponding to the non-negative eigenvalues $\{E_j(t)\}_{j\geq 1}$
counting multiplicities and arranged in decreasing order. Then using $T'(t)=\Delta T$ we obtain: 
\begin{align}
\phi'(t)&=\sum_{j\geq 1}\langle f_j(t), (\Delta T) f_j(t)\rangle
\frac{1}{2\pi i}\int_\Gamma F(z)(z-E_j(t))^{-2} dz\nonumber \\
&=\sum_{j\geq 1}\langle f_j(t), (\Delta T) f_j(t)\rangle F'(E_j(t))={\rm Tr}\{F'(T(t))\Delta T\}.
\end{align}
Thus $\phi'(t)={\rm Tr}\{(1+T(t))^{-1}\Delta T\}$, which means that $|\phi'(t)|\leq ||\Delta T||_1$ and this proves 
\eqref{septemb1}. 

We now prove \eqref{septemb2}. Note the identity: 
$$\phi(1)=\phi(0)+\phi'(0)+\int_0^1dt\; (1-t)\phi''(t).$$
Clearly, $\phi'(0)={\rm Tr}\{F'(T_1)\Delta T\}$; thus the only remaining thing is to show that $\phi''(t)\leq 0$ for all $t$. 
By differentiating once again in \eqref{iulie3} we have:
\begin{align}\label{iulie4}
\phi''(t)
&=\frac{1}{2\pi i}\int_\Gamma F'(z){\rm Tr}
\{(z-T(t))^{-1}\Delta T (z-T(t))^{-1}\Delta T\}dz.
\end{align}
Using the eigenvalues and eigenvectors of $T(t)$ we get:
\begin{align}\label{iulie5}
\phi''(t)
&=\frac{1}{2\pi i}\sum_{j\geq 1} \sum_{k\geq 1}|\langle f_j(t), \Delta T f_k(t)\rangle |^2\int_\Gamma F'(z)
(z-E_j(t))^{-1}(z-E_k(t))^{-1}dz.
\end{align}
Define $\{a_{jk}\}_{j,k\geq 1}$, where $a_{jk}(t):=F''(E_j(t))$
if $E_j(t)=E_k(t)$, and
$a_{jk}(t)=\frac{F'(E_j(t))-F'(E_k(t))}{E_j(t)-E_k(t)}$ if $E_j(t)\neq
E_k(t)$. Because $F$ is concave on $(0,\infty)$, all $a_{jk}$'s are non-positive. By the residue calculus we have:
\begin{align}\label{iulie6}
\phi''(t)
&=\sum_{j\geq 1} \sum_{k\geq 1}|\langle f_j(t), \Delta T  f_k(t)\rangle|^2a_{jk}(t)\leq 0,\quad 0<t<1,
\end{align}
thus the proof of \eqref{septemb2} is over.
\end{proof}

\vspace{0.5cm}

\noindent {\bf Remark}. Assume that $P=P^*=P^2$ is an orthogonal projection, and denote by $Q=1-P$. Define 
$T_1=PT_2P+QT_2Q$ to be the 'diagonal' part of $T_2$ with respect to the decomposition $P+Q=1$. Then $T_2-T_1$ is off-diagonal 
and ${\rm Tr}\{F'(T_1) \Delta T\}=0$. Then \eqref{septemb2} implies 
$$  {\rm Tr}\{F(T_2)\}\leq {\rm Tr}\{F(T_1)\}= {\rm Tr}\{F(PT_2P)\}+{\rm Tr}\{F(QT_2Q)\}$$
which is a variant of Berezin's inequality.

\subsection{Proof of (i)}
The operators $A_{T/R}$ (see \eqref{iulie11-4}) have each a $6\times 6$ matrix valued integral kernel $A_{T/R}(\Omega,\Omega')$ 
 with the following structure:
\begin{align}\label{iulie11-11}
&A_{T}^{ij}(\Omega,\Omega')=\sum_{m=1}^{M_{T}}a_{T}^{im}(\Omega)\overline{a_{T}^{jm}(\Omega')}, \nonumber \\
&A_{R}^{ij}(\Omega,\Omega')=\sum_{m=1}^{M_{R}}\overline{a_{R}^{mi}(\Omega)}a_{R}^{jm}(\Omega'),\quad i,j\in\{1,2,\dots,6\}.
\end{align}
A consequence of \eqref{iulie11-12} is the following estimate (in the sense 
of bounded operators generated by the corresponding integral kernels): 
\begin{equation}\label{septemb5}
||A_T/M_T-\mathcal{A}_T||=\mathcal{O}(M_T^{-1}),\quad ||\sqrt{A_T/M_T}-\sqrt{\mathcal{A}_T}||={o}(1).
\end{equation}
Note that in general we cannot say more about the speed of convergence of the square root. 

Remember from Lemma \ref{lemma2} that $\mathcal{K}=
\sqrt{A_T} \mathcal{S}^*A_R \mathcal{S} \sqrt{A_T}$. 
Identify $E_T\mathcal{K}/(N_0 M_T)$ with $T_1$, and $E_T\sqrt{\mathcal{A}_T} \mathcal{S}^*A_R \mathcal{S} 
\sqrt{\mathcal{A}_T}/N_0$ with $T_2$. Since $\mathcal{S}$ is assumed to be Hilbert-Schmidt, we have that $T_2-T_1$ is trace class
with a trace norm which goes to zero with $M_T$. Thus \eqref{septemb1} implies:
\begin{equation}\label{septemb6}
\lim_{M_T\to\infty}C(M_T,M_R)=\frac{1}{\rm ln(2)} {\rm Tr}_{[L^2(\mathbb{S}^2)]^6}{\rm ln} \left (1+\frac{E_T}{ N_0}\sqrt{\mathcal{A}_T} \mathcal{S}^*A_R \mathcal{S} 
\sqrt{\mathcal{A}_T}\right )
\end{equation}
and Theorem \ref{theo} {\it (i)} is proved.

\subsection{Proof of (ii)}

Let us 
introduce the operator 
$$\mathcal{D}_R:=\frac{E_T}{M_T N_0}\sqrt{A_T}\mathcal{S}^*\mathcal{A}_{R}\mathcal{S}\sqrt{A_{T}}.$$ 
Another consequence of \eqref{iulie11-12} and \eqref{iulie11-11} is (again as bounded operators):
\begin{equation}\label{septemb7}
\sup_{M_R\geq 1}||A_R-M_R\mathcal{A}_R||<\infty.
\end{equation}
Identify $E_T\mathcal{K}/(N_0 M_T)$ with $T_1$, and $M_R \mathcal{D}_R$ with $T_2$. Then $||T_2-T_1||_1$ is uniformly bounded in 
$M_R$, thus the estimate \eqref{septemb1} implies:
\begin{equation}\label{septemb8}
\sup_{M_R\geq 1}\left |C(M_T,M_R)-\frac{1}{\rm ln(2)} {\rm Tr}_{[L^2(\mathbb{S}^2)]^6}{\rm ln} 
\left (1+M_R \mathcal{D}_R\right )\right |<\infty.
\end{equation}
 But now $\mathcal{D}_R$ is a rank $M_T$, non-negative operator, and it is non-zero only on the range of $A_T$. Assume that 
$\mathcal{D}_R$ has exactly $d$ positive eigenvalues, denoted by $\{\lambda_j\}_{j=1}^d$, including multiplicities. Then 
$$\lim_{M_R\to\infty}\frac{1}{\ln(M_R)}{\rm Tr}_{[L^2(\mathbb{S}^2)]^6}{\rm ln} 
\left (1+M_R \mathcal{D}_R\right )=\lim_{M_R\to\infty}\frac{1}{\ln(M_R)}\sum_{j=1}^d{\rm ln} 
\left (1+M_R \lambda_j\right ) =d\leq M_T,$$
and the proof of Theorem \ref{theo} {\it (ii)} is over.

\subsection {Proof of (iii)}

Let us now examine the situation in which the operator generated by the spread function has finite rank. 
This would be the case if we have 
$N$ isolated scatterers in the environment. Then the operator $\mathcal{S}$ whose kernel is given by the spread function  
can be written in the form 
\begin{equation}\label{septemb100}
\mathcal{S}=\sum_{j,k=1}^N c_{jk} |f_j\rangle \langle g_k|,
\end{equation}
where $c_{jk}$'s are complex numbers, while $\{|f_j\rangle\}_{j=1}^N$ and $\{|g_j\rangle\}_{j=1}^N$ are 
(not necessarily unit) vectors in $[L^2(\mathbb{S}^2)]^6$. 
Remember that $M_T=M$ and $M_R=aM$, with $a>0$ a constant. Let us introduce the operator 
\begin{equation}\label{septemb10}
{D}_M:=\frac{E_T}{a N_0}\sqrt{A_T/M}\mathcal{S}^*({A}_{R}/M)\mathcal{S}\sqrt{A_{T}/M},
\end{equation}
thus:
\begin{equation}\label{septemb11}
C(M,aM)=\frac{1}{\rm ln(2)} {\rm Tr}_{[L^2(\mathbb{S}^2)]^6}{\rm ln} 
\left (1+aM{D}_M\right ).
\end{equation}
The operator ${D}_M$ can be written as:
$${D}_M=\frac{E_T}{a N_0}\sum_{j,k=1}^N \left\langle \sum_{i=1}^Nc_{ij}f_i, ({A}_{R}/M)\sum_{m=1}^Nc_{mk}f_m\right \rangle \:
|\sqrt{{A}_{T}/M}g_j\rangle \langle \sqrt{{A}_{T}/M} g_k|.$$
If we denote by:
\begin{equation}\label{iulie11-14}
d_{jk}:=\frac{E_T}{a N_0}\left\langle \sum_{i=1}^Nc_{ij}f_i, ({A}_{R}/M)\sum_{m=1}^Nc_{mk}f_m\right \rangle,\quad 
|h_j\rangle :=|\sqrt{{A}_{T}/M} \;g_j\rangle,
\end{equation}
then 
\begin{equation}\label{iulie11-15}
{D}_M=\sum_{j,k=1}^N d_{jk}|h_j\rangle\langle h_k|.
\end{equation}
We begin with a lemma:
\begin{lemma}\label{lemma10}
The matrix $d$ is non-negative, and we have the following limit:
$$\lim_{M\to\infty}d_{jk}=:\tilde{d}_{jk}=\frac{E_T}{ N_0}\left\langle \sum_{i=1}^Nc_{ij}f_i, \mathcal{A}_{R}
\sum_{m=1}^Nc_{mk}f_m\right \rangle.$$
\end{lemma}
\begin{proof}
The convergence is implied by \eqref{septemb5}, so we only need to prove non-negativity of $d$. Choose any 
 $\{\psi_j\}_{j=1}^N\in\mathbb{C}^N$ and define 
$\Psi:=\sum_{i=1}^N\sum_{j=1}^Nc_{ij}\psi_j |f_i\rangle \in [L^2(\mathbb{S}^2)]^6$. 
Then compute:
\begin{align}
\langle \psi, d\;\psi\rangle_{\mathbb{C}^N}=\sum_{j=1}^N\sum_{k=1}^N\overline{\psi_j}d_{jk}\psi_k=
\frac{E_T}{aM N_0}\left\langle \Psi, {A}_{R}\Psi\right \rangle\geq 0,
\end{align}
where we used the non-negativity of ${A}_{R}$, see Lemma \ref{lemma1}.
\end{proof}

\vspace{0.5cm}

\begin{lemma}\label{lemma11}
Let $\phi_{jk}:=\langle h_j,h_k\rangle$ denote an $N\times N$ matrix constructed 
with the vectors introduced in \eqref{iulie11-14}.
Then $\phi$ is non-negative, and we have the following limit:
$$\lim_{M\to\infty}\phi_{jk}=:\tilde{\phi}_{jk}=\langle g_j, \mathcal{A}_{T} g_k \rangle.$$
\end{lemma}
\begin{proof}
The convergence is implied by \eqref{septemb7}, while the non-negativity of $\phi$ follows in the same way as for $d$, 
using the non-negativity of ${A}_{T}$ proved in Lemma \ref{lemma1}. We do not give further details.
\end{proof}

\vspace{0.5cm}

The following result expresses the capacity as a determinant of a finite rank matrix, uniformly in $M$.
\begin{lemma}\label{lemma4}  We have the identity:
\begin{equation}
 C(M,aM)=\frac{1}{\ln(2)} {\rm ln}\;{\rm det}\left (1+aM\;\sqrt{d} \phi \sqrt{d}\right ).
\end{equation} 
\end{lemma}
\begin{proof}
It is enough to prove the equality 
\begin{equation}\label{iulie11-17}
 {\rm Tr}_{[L^2(\mathbb{S}^2)]^6}\;{\rm Ln} \left (1+zD_M\right )={\rm Tr}_{\mathbb{C}^N}\;{\rm Ln}
\left (1+z\sqrt{d} \phi \sqrt{d}\right )
\end{equation}
for any $z$ with $|z|$ sufficiently small. Then since both sides of \eqref{iulie11-17} are analytic 
functions in the half plane ${\rm Re}(z)>0$, the equality will also hold for  $z=a M$. 
 We will show that both sides of \eqref{iulie11-17} are given by the same power series around $z=0$, which amounts to 
proving that ${\rm Tr}_{[L^2(\mathbb{S}^2)]^6}\;{D}_M^p={\rm Tr}_{\mathbb{C}^N}\;(\sqrt{d} \phi \sqrt{d})^p$ for any $p\geq 1$. 
This is in fact a direct consequence of the identity (easily provable by induction)
$${D}_M^p=\sum_{j,k=1}^N [d(\phi d)^{p-1}]_{jk}|h_j\rangle\langle h_k|,\quad p\geq 1,$$
in which one has to take the trace and use its cyclicity in order to move a $\sqrt{d}$ from left to the right. The proof is over. 
\end{proof}

\vspace{0.5cm}

\begin{lemma}\label{lemma13}  The capacity grows at most like a logarithm:
\begin{equation}
 \limsup_{M\to\infty}\frac{C(M,aM)}{\ln(M)}\leq \frac{N}{\ln(2)} <\infty.
\end{equation} 
\end{lemma}
\begin{proof}
Seen as an operator on $\mathbb{C}^N$, $\sqrt{d} \phi \sqrt{d}$ converges in operator norm to 
$\sqrt{\tilde{d}} \;\tilde{\phi} \;\sqrt{\tilde{d}}$ when $M$ grows. Due to regular perturbation theory, it follows that all $N$  
eigenvalues of $\sqrt{d} \phi \sqrt{d}$ are uniformly bounded in $M$, say by a number $\lambda>0$. Thus 
\begin{equation}\label{septemb101}
C(M,aM)\leq \frac{N}{\ln(2)}\; \ln(1+aM\lambda)\leq  \frac{N}{\ln(2)}\;\{\ln(M)+ \ln(1+a\lambda)\},\quad M\geq 1,
\end{equation}
which ends the proof of Theorem \ref{theo} {\it (iii)}.  
\end{proof}

\subsection{Proof of (iv)}
Remember that here we no longer demand $\mathcal{S}$ to have finite rank, but we assume that it has a smooth integral kernel. 
If $L_B$ denotes the usual (non-negative) Laplace-Beltrami operator densely defined in $L^2(\mathbb{S}^2)$, 
then we denote by $\tilde{L}_B:=\oplus_{j=1}^6L_B$ the corresponding operator in $[L^2(\mathbb{S}^2)]^6$. The smoothness of 
$s(\Omega,\Omega')$ implies that for every natural number $n$, the operators 
$$\mathcal{S}_n:=\tilde{L}_B^n\mathcal{S},\quad \mathcal{S}_n^*=\mathcal{S}^*\tilde{L}_B^n$$
are Hilbert-Schmidt operators. We know that $L_B$ has purely discrete spectrum and the distribution of eigenvalues obeys 
Weyl's law. More precisely, for every $E>0$ define $P_E$ and $\tilde{P}_E$ to be the projectors corresponding to the 
eigenvalues of $L_B$ and respectively $\tilde{L}_B$ 
which are less or equal than $E$. Then it is well known \cite{Horm} that:
\begin{equation}\label{sept11_22}
 \lim_{E\to\infty}\frac{{\dim}(P_E)}{E}=1,\quad \lim_{E\to\infty}\frac{{\dim}(\tilde{P}_E)}{E}=6.
\end{equation}

The starting point of our proof are formulas \eqref{septemb10} and \eqref{septemb11}. Introduce
\begin{equation}
T_1=aM{D}_M,\quad T_2=aM\frac{E_T}{a N_0}\sqrt{A_T/M}\mathcal{S}^*\tilde{P}_E({A}_{R}/M)\tilde{P}_E\mathcal{S}\sqrt{A_{T}/M},
\end{equation}
where $T_2$ is obtained by inserting two projections $\tilde{P}_E$ inside $D_M$. 
Then it is easy to see that there exists a constant $C_1>0$ independent of $M$ and $E$ such that 
\begin{equation}
||T_1- T_2||_1\leq C_1\; M ||(1-\tilde{P}_E)\mathcal{S}||_2,\quad M\geq 1.
\end{equation}
But $$||(1-\tilde{P}_E)\mathcal{S}||_2=||(1-\tilde{P}_E)\tilde{L}_B^{-n}\mathcal{S}_n||_2\leq 
E^{-n}||\mathcal{S}_n||_2$$
for all $E>0$, thus 
\begin{equation}\label{septemb21}
||T_1- T_2||_1\leq C_1\; M E^{-n}||\mathcal{S}_n||_2,\quad M\geq 1,\quad E>0.
\end{equation}
The crucial observation is that $\tilde{P}_E\mathcal{S}$ has finite rank, equal to $N={\rm dim}(\tilde{P}_E)$, 
and the method of Theorem \ref{theo} {\it (iii)} can be applied. The only problem is that we cannot be sure that the bound 
$\lambda$ in \eqref{septemb101} can be chosen independent of $N$, but we will now show that in the worst case scenario $\lambda$ 
grows proportionally with $N^2$. 

\begin{lemma}\label{lemma20}
Let $N={\rm dim}(\tilde{P}_E)$. There exists another constant $C_2$ independent of 
$M$ and $N$ such that:
\begin{equation}\label{lemma200}
{\rm Tr}\{{\rm ln}(1+T_2)\}\leq N \ln(1+C_2MN^2).
\end{equation}
\end{lemma}
\begin{proof}
We express $\tilde{P}_E$ as $\sum_{j=1}^N|\psi_j\rangle\langle\psi_j|$, where the 
$\psi_j$'s are normalized eigenfunctions of $\tilde{L}_B$ spanning the range of $\tilde{P}_E$. Then 
$$\tilde{P}_E\mathcal{S}=\sum_{j=1}^N|\psi_j\rangle\langle \mathcal{S}^*\psi_j|.$$
Comparing with \eqref{septemb100} we see that $f_j=\psi_j$, $c_{jk}=\delta_{jk}$ and $g_j=\mathcal{S}^*\psi_j$. Looking at the 
definition of $d_{jk}$ in \eqref{iulie11-14} we see that $|d_{jk}|$ is bounded uniformly in $j$, $k$ and $M$. Thus the norm 
of the matrix $d$ can grow at most as $N$. The same conclusion holds for the matrix $\phi$ defined in Lemma \ref{lemma11}. 
It means that the norm of $\sqrt{d}\phi \sqrt{d}$ (which is equal to its largest eigenvalue) grows at most as $N^2$. Thus we 
can choose some $\lambda =C_3N^2$ with $C_3$ a constant independent of $M$ and $N$ and use it in \eqref{septemb101}. The proof 
is over.
\end{proof}

\vspace{0.5cm}

Now using \eqref{lemma200}, \eqref{septemb21} and \eqref{sept11_22} in \eqref{septemb1} we obtain that for every $n\geq 1$ there exists a constant 
$K_n>1$ independent of $M$ and $E$ such that:
$$C(M,aM)\leq K_n\left ( E \ln(1+K_n M\; E^2)+ M E^{-n}\right ).$$
Now fix an $0<\epsilon<1$. Choose $E=M^{\frac{1-\epsilon}{n}}$ and introduce it in the above estimate. Since $K_n M\; E^2>1$, we 
have 
$$\ln(1+K_n M\; E^2)=\ln(1+K_n^{-1} M^{-1}\; E^{-2})+ \ln(K_n M\; E^2)\leq \ln(2)+\ln(K_n)+(1+\frac{2-2\epsilon}{n})\ln(M).$$  
Thus we get another constant $\tilde{K}_n>1$ such that uniformly in $M>1$ and $0<\epsilon<1$ we have:
 $$C(M,aM)\leq \tilde{K}_n\left (1+ M^{\frac{1-\epsilon}{n}} \ln(M)+ M^\epsilon\right ).$$
But now we can choose $n_\epsilon$ to be the smallest natural number such that $\frac{1-\epsilon}{n}\leq \epsilon/2$. Thus 
$\sup_{M\geq 1}\frac{C(M,aM)}{M^\epsilon}<\infty$ and the proof of
Theorem \ref{theo} is over.

\section{Conclusions}
In the case when the antennas occupy a given volume, our mathematical
results will not constitute a big surprise for the engineers and researchers
who have been involved in this type of MIMO studies and who have also
predicted a sublinear behavior, even though they used ad-hoc statistical models for
the transfer matrix. In this scenario, both our
deterministic model and the stochastic ones seem to predict that the
capacity growth can no longer be
considered as linear if the number of antennas passes over a relatively low
threshold. 

If the distance in between the antennas is maintained constant,
the situation is rather different. All standard statistical models
predict a linear increase in this case. But in a forthcoming paper we will
apply our deterministic model in order to confirm the results of 
\cite{BS} which predicted a sublinear behavior even in this
scenario. We will show that the sublinear growth begins from a not so
large threshold value of the number of antennas, and that the
reachness of the scattering environment is at least as important as
the number of antennas.  This shows that the discussion on the models is important.

\section{Acknowledgments}

The authors acknowledge support from the Danish FNU grant 
{\it Mathematical Physics}. We also thank B. Fleury for drawing to our attention the reference \cite{Po}.


\begin{thebibliography}{}
\bibitem{BCFM}F.Bentosela, H.D. Cornean, B.Fleury, N. Marchetti 
On the tranfer matrix of a MIMO system.
Math. Meth. Appl. Sci. {\bf 34}(8), 963–976 (2011)

\bibitem{Fo} G. Foschini, M. Gans: On limits of wireless communications in a fading environment when using
multiple antennas. Wireless Pers.Comm (London) {\bf 6} (3), 311-335 (1998).
\bibitem{Te} I.E.Telatar: Capacity of multiantenna Gaussian
  channel. Eur.Trans. Commun. {\bf 10} (6), 585-595 (1999)
\bibitem{P}A.Paulraj, R.Nabar and D. Gore: 
Introduction to Space-Time Wireless Communications , 
Cambridge University Press, Cambridge , UK, 2003
\bibitem{M}A.F. Molisch, : Wireless Communications 
Wiley-IEEE Press, New-York,NY,USA,2005
\bibitem{Al}P.Almers, E.Bonek, A.Burr, N.Czink, M.Debbah ,
V Degli-Esposti, H.Hofstetter, P.Kysti, D. Laurenson, G. Matz, A.F. Molisch, C. Oestges and H.Ozcelik :
 Survey of Channel and Radio Propagation Models for Wireless MIMO Systems ,
 EURASIP Journal on Wireless Communications and Networking
 Volume 2007, Article ID 19070, 19 pages

\bibitem{G1} D. Gesbert, T. Ekman, N. Christophersen: Capacity limits of dense palm-sized MIMO arrays. IEEE Global Telecommunications 
Conference {\bf 2} 1187-1191 (2002)

\bibitem{W} S. Wei, D. Goeckel, R. Janaswamy: On the Asymptotic Capacity of MIMO Systems with Antenna Arrays of Fixed Length. IEEE Transactions on 
Wireless Communications {\bf 4}(4) (2005)
 
 \bibitem{Po} A.S.Y. Poon, R.W. Brodersen, D.N.C. Tse: Degrees of Freedom in Multiple-Antenna
Channels: A Signal Space Approach. IEEE Transactions on Information
Theory {\bf 51} (2), 523-536 (2005)




\bibitem{Horm} L. H\"ormander: The analysis of linear partial
  differential operators IV: Fourier integral operators, 
Springer-Verlag Berlin Heidelberg New York (1985)

\bibitem{BS}
F.Bentosela, E.Soccorsi: Sub-linear capacity scaling for multi-path
channel models, Math. Methods Appl. Sci. {\bf 33} (9), 1164—1180 (2010)






 


 
\end{thebibliography}
\end{document}